\documentclass{elsarticle}
\usepackage{amssymb,amsmath,amsthm}
\usepackage{graphicx,enumerate,fullpage,hyperref}
\usepackage{amsfonts,amssymb,url}
\usepackage{epsfig,enumerate,color}
\usepackage{algorithm}
\usepackage{algorithmic}

\newtheorem{theorem}{Theorem}
\newtheorem{lemma}{Lemma}

\newtheorem{proposition}{Proposition}

\begin{document}

\begin{frontmatter}
\title{On the complexity of the vector connectivity problem}

\author[SAL]{Ferdinando Cicalese}
\ead{cicalese@dia.unisa.it}
\author[MM-IAM,MM-FAMNIT,MM-IMFM]{Martin Milani\v c
}
\ead{martin.milanic@upr.si}
\author[VER]{Romeo Rizzi}
\ead{romeo.rizzi@univr.it}


\address[SAL]{Department of Computer Science, University of Salerno,  Italy}
\address[MM-IAM]{University of Primorska, UP IAM, Muzejski trg 2, SI-6000 Koper, Slovenia}
\address[MM-FAMNIT]{University of Primorska, UP FAMNIT, Glagolja\v{s}ka 8, SI-6000 Koper, Slovenia}
\address[MM-IMFM]{IMFM, Jadranska 19, 1000 Ljubljana, Slovenia}
\address[VER]{Department of Computer Science, University of Verona, Italy}

\begin{abstract}
\begin{sloppypar}
We study a relaxation of the {\sc Vector Domination} problem called {\sc Vector Connectivity} ({\sc VecCon}).
Given a graph $G$ with a requirement $r(v)$ for each vertex $v$, {\sc VecCon} asks for a minimum cardinality set of vertices $S$ such that
every vertex $v\in V\setminus S$  is connected to $S$ via $r(v)$ disjoint paths.
In the paper introducing the problem, Boros et al.~[Networks, 2014, to appear] gave polynomial-time solutions for {\sc VecCon} in trees, cographs, and split graphs, and showed that the problem can be approximated in polynomial time on $n$-vertex graphs to within a factor of $\log n+2$, leaving open the question of whether the problem is {\sf NP}-hard on general graphs. We show that {\sc VecCon} is {\sf APX}-hard in general graphs, and {\sf NP}-hard in planar bipartite graphs and in planar line graphs.
We also generalize the polynomial result for trees by solving the problem for block graphs.
\end{sloppypar}
\end{abstract}

\begin{keyword}
vector connectivity\sep {\sf APX}-hardness\sep {\sf NP}-hardness\sep polynomial-time algorithm\sep block graphs\\
\MSC[2010] 05C40 \sep  05C85\sep 90C27 \sep 68Q25
\end{keyword}

\end{frontmatter}

\section{Introduction and background}

Recently, Boros et al.~\cite{BHHM} introduced the {\sc Vector Connectivity}  problem ({\sc VecCon}) in graphs. This problem takes as input a graph $G$ and an integer $r(v)\in \{0,1,\ldots, d(v)\}$ for every vertex $v$ of $G$, and the objective is to find a vertex subset $S$ of minimum cardinality such that every vertex $v$ either belongs to $S$, or is connected to at least $r(v)$ vertices of $S$ by disjoint paths.
If we require each path to be of length exactly $1$, we get the well-known {\sc Vector Domination} problem~\cite{MR1741406}, which is a generalization of the famous {\sc Dominating Set} and {\sc Vertex Cover} problems.

The {\sc Vector Connectivity}  problem is one of the several problems related to domination and connectivity, which have seen renewed attention in the last few years in connection with the flourishing area of
information spreading 
(see, e.g., \cite{MR1741406,Kempe:2005,Cic-Mil-Vac:2013,Dinh:2014,Ackerman:2010,Chopin:2014,Cicalese201365} and
references therein quoted).
In a viral marketing campaign one of the problems is to identify a set of targets in a (social) network
that can be influenced (e.g., on the goodness of a product) and such that from them most/all
the network can be influenced (e.g., convinced to buy the product).
The model is based on the assumption that each vertex has a threshold $r(v)$
such that when $r(v)$ neighbors are influenced, also $v$ will get convinced too.
Assume now that for  $v$ it is not enough that $r(v)$ neighbors are convinced about the product.
Vertex $v$ also requires that their motivations are independent.
A way to model this ``skeptic'' variant of influence spreading is to require that each vertex in the network must be reached by $r(v)$
vertex-disjoint paths originating in the target set.

Another scenario where the vector connectivity problem arises is the following: Each vertex in a network produces a certain amount of a given good. We want to place in the network warehouses where the good can be stored.
For security/resilience reasons it is better if from each source to each destination (warehouse)
only a small amount of the good (e.g., one unit) travels at once.
In particular, it is preferred if the units of good  from one location to the different warehouses travel on different routes.
This reduces the risk that if delivery gets intercepted or attacked or interrupted by a fault on the network
a large amount of the good gets lost. It is not hard to see that finding the minimum number of warehouses
given the amount of units produced at each vertex coincides with the vector connectivity problem.

Boros et al.~developed polynomial-time algorithms for {\sc VecCon} on split graphs, cographs, and trees, and showed how to model the problem as a minimum submodular cover problem, leading to a polynomial-time algorithm approximating {\sc VecCon} within a factor of $\ln n +2$ on all $n$-vertex graphs. One of the questions left open in that paper was whether on general graphs, {\sc VecCon} is polynomially solvable or {\sf NP}-hard.
In this paper, we answer this question by showing that {\sc VecCon} is {\sf APX}-hard (and consequently {\sf NP}-hard) in general graphs.
Our reduction is from the {\sc Vertex Cover} problem in cubic graphs. Simple modifications of the hardness proof allow us to also show that
{\sc VecCon} is hard in several graph classes for which the {\sc Vertex Cover} problem is polynomially solvable, such as
bipartite graphs and line graphs.

Our hardness results remain valid for input instances in which the vertex requirements $r(v)$ are bounded by $4$.
On the other hand, we show that {\sc VecCon} can be solved in polynomial time for requirements bounded by $2$, thus leaving open only the case
with maximum requirement $3$. We also develop a polynomial-time solution for {\sc VecCon} in block graphs, thereby generalizing the result
by Boros et al.~\cite{BHHM} showing that {\sc VecCon} is polynomial on trees. This result is obtained by introducing a more general problem called {\sc Free-Set Vector Connectivity} ({\sc FreeVecCon} for short), and developing an algorithm that reduces an instance of {\sc FreeVecCon} to
solving instances of {\sc FreeVecCon} on  biconnected components of the input graph.

Vertex covers and dominating sets in a graph $G$ can be easily characterized as hitting sets of derived hypergraphs
(of $G$ itself, and of the closed neighborhood hypergraph of $G$, respectively). We give a similar characterization of
vector connectivity sets.

The paper is structured as follows. In Section~\ref{sec:prelim}, we collect all the necessary definitions.
In Section~\ref{sec:char-Menger}, we develop a characterization of vector connectivity sets as hitting sets of a derived hypergraph, which is
followed with hardness results in Section~\ref{sec:hard}.
A polynomial reduction for a problem generalizing {\sc VecCon} to biconnected graphs is given
in Section~\ref{sec:reduction}. Some of the algorithmic consequences of this reduction are examined in
Section~\ref{sec:poly}. We conclude the paper with some open problems in Section~\ref{sec:concl}.

\subsection{Definitions and Notation}\label{sec:prelim}

All graphs in this paper are simple and undirected, and will be denoted by $G=(V,E)$, where $V$ is the set of vertices and $E$ is the set of edges.
We use standard graph terminology. In particular, the degree of a vertex $v$ in $G$ is denoted by $d_G(v)$,
the neighborhood and the closed neighborhood of a vertex $v$ are denoted by $N_G(v)$ and $N_G[v]$, respectively,
and $V(G)$ refers to the vertex set of $G$. Moreover, for a set $X\subseteq V(G)$, we define $N_G(X) = \left(\cup_{v\in X}N_G(v)\right)\setminus X$ and
denote by $G[X]$ the subgraph of $G$ induced by $X$.
Given a graph $G=(V,E)$, a set $S\subseteq V$ and a vertex $v\in V\setminus S$, a {\em $v$--$S$ fan of order $k$} is a collection of $k$ paths $P_1,\ldots, P_k$ such that (1) every $P_i$ is a path connecting $v$ to a vertex of $S$, and (2) the paths are pairwise vertex-disjoint except at $v$, i.e., for all $1\le i<j\le k$, it holds that $V(P_i)\cap V(P_j) = \{v\}$.

Given a graph $G = (V,E)$ and an integer-valued function $r:V\to\mathbb{Z}_+ = \{0,1,2,\ldots\}$, a {\it vector connectivity set} for $(G,r)$ is a set $S\subseteq V$ such that there exists a $v$--$S$ fan of order $r(v)$ for every $v\in V\backslash S$. We say that $r(v)$ is the {\em requirement} of vertex $v$. The {\sc VecCon} problem is the problem of finding a vector connectivity set of minimum size for $(G,r)$.
The minimum size of a vector connectivity set for $(G,r)$ is denoted by $\kappa(G,r)$.
In Boros et al.~\cite{BHHM}, it was assumed that vertex requirements do not exceed their degrees.
Since our polynomial results are developed using the more general variant of the problem, we do not impose this restriction.
At the same time, our hardness results also hold for the original, more restrictive variant.

For every $v\in V$ and every set $S\subseteq V\setminus \{v\}$, we say that $v$ is {\em $r$-linked} to $S$ if there is a $v$--$S$ fan of order $r$ in $G$. Hence, given an instance $(G,r)$ of {\sc VecCon}, a set $S\subseteq V$ is a vector connectivity set for $(G,r)$ if and only if every $v\in V\setminus S$ is $r(v)$-linked to~$S$. Given a vertex requirement function $r:V\to\mathbb{Z}$ and a non-empty set $X\subseteq V$, we define
$R(X):= \max_{x\in X}r(x)$. A graph $G$ is {\it $k$-connected} if $|V(G)|>k$ and for every $S\subseteq V(G)$ with $|S|<k$, the
graph $G-S$ is connected.

\section{A characterization of vector connectivity sets}\label{sec:char-Menger}

Menger's Theorem~\cite{zbMATH02582908} implies the following
characterization of vector connectivity sets,
showing that they are exactly
the hitting sets of a certain hypergraph derived from graph $G$ and vertex requirement function $r$.
This characterization will be used in our proof of Theorem~\ref{thm:NP-c}
in Section~\ref{sec:hard}.

\begin{sloppypar}
\begin{proposition}\label{prop:char-Menger}
For every graph $G = (V,E)$, vertex requirements $r:V\to\mathbb{Z}_+$, and a set $S\subseteq V$, the following conditions are equivalent:
\vspace{-0.2cm}
\begin{enumerate}[(i)]
  \item $S$ is a vector connectivity set for $(G,r)$.
  \item For every non-empty set $X\subseteq V$ such that $G[X]$ is connected and \hbox{$R(X)>|N_G(X)|$}, we have $S\cap X\neq \emptyset$.
\end{enumerate}
\end{proposition}
\end{sloppypar}

\begin{proof}
First, let $S$ be a vector connectivity set for $(G,r)$.
Suppose for a contradiction that there exists
a non-empty set $X\subseteq V$ such that
$G[X]$ is connected, $R(X)>|N_G(X)|$, and $S\cap X = \emptyset$.
Let $C = N_G(X)$, and let $x\in X$ be a vertex such that $r(x)>|C|$.
Since $S\cap X= \emptyset$, we have $x\not\in S$.
Moreover, the definition of $C$ implies that in the graph $G-C$,
there is no path from $x$ to $S$.
Therefore, by Menger's Theorem, the maximum number of disjoint $x$-$S$ paths
is at most $|C|$, contrary to the fact that $x$ is $r(x)$-linked to $S$
and $r(x)>|C|$.

Conversely,
suppose for a contradiction that $S\subseteq V$ is not a vector connectivity set for $(G,r)$, and
that for every non-empty set $X\subseteq V$ such that
$G[X]$ is connected and
$|N_G(X)|<R(X)$, we have $S\cap X\neq \emptyset$.
Since $S$ is not a vector connectivity set for $(G,r)$, there exists a vertex $x\in V\setminus S$ such that $x$ is not $r(x)$-linked to $S$.
By Menger's Theorem, there exists a set $C\subseteq V\setminus\{x\}$ such that $|C|<r(x)$ and
every path connecting $x$ to $S$ contains a vertex of $C$.
Let $X$ be the component of $G-C$ containing $x$.
Then, $G[X]$ is connected and $N_G(X)$ is contained in $C$, implying
$|N_G(X)|\le|C|< r(x)\le R(X)$. Hence, by the assumption on $S$, we have $S\cap X\neq \emptyset$.
But this means that there exists a path connecting $x$ to $S$ avoiding $C$, contrary to the choice of $C$.
\end{proof}

\section{Hardness results}\label{sec:hard}

We start with the {\sf NP}-hardness results.

\begin{theorem}\label{thm:NP-c}
The decision version of the {\sc VecCon} problem restricted to instances with maximum requirement $4$ is {\sf NP}-complete, even for:
\begin{itemize}
  \item $2$-connected planar bipartite graphs of maximum degree $5$ and girth at least $k$ (for every fixed $k$),
  \item $2$-connected planar line graphs of maximum degree $5$.
\end{itemize}
\end{theorem}

\begin{proof}
Membership in {\sf NP} follows from the fact that the feasibility of a solution $S$ can be tested in polynomial time,
using, e.g., Menger's Theorem and max flow algorithms.

We first show hardness of the problem for $2$-connected planar line graphs of maximum degree $5$, and then show how to modify the construction to
obtain a $2$-connected planar bipartite graph of maximum degree $5$ and girth at least $k$ (for a fixed $k$).

The hardness reduction is from {\sc Vertex Cover} in $2$-connected cubic planar graphs, a problem shown {\sf NP}-complete by Mohar~\cite{MR1828438}.
Recall that in the {\sc Vertex Cover} problem, the input is a graph $G$ and an integer $k$, and the task is to determine whether $G$ contains a vertex cover of size at most $k$, where a {\it vertex cover} in a graph is a set of its vertices such that every edge of the graph has an endpoint in the set.

Suppose that $(G,k)$ is an instance to the {\sc Vertex Cover} problem such that $G$ is a $2$-connected cubic planar graph.
\begin{sloppypar}
We construct a graph $G' = (V', E')$ in 3 steps, using the following procedure:
\begin{enumerate}
   \item Replace each edge $e = xy$ of $G$ with a path on $5$ vertices $(x, w_{x,e}, w_e, w_{y,e}, y)$.
   Formally, delete edge $xy$, add three new vertices $w_{x,e}, w_e, w_{y,e}$ and
  edges $xw_{x,e}$, $w_{x,e}w_e$, $w_ew_{y,e}$, $w_{y,e}y$.

  \item For each edge $e$ of $G$, glue a triangle on top of $w_{x,e}w_e$ with tip $z_{x,e}$ (a new vertex),
  and a triangle on top of $w_{y,e}w_e$ with tip $z_{y,e}$ (a new vertex),
  Formally, add two new vertices $z_{x,e}$ and $z_{y,e}$ and edges $w_{x,e}z_{x,e}$, $z_{x,e}w_e$,
  $w_ez_{y,e}$, $z_{x,e}w_{y,e}$.

  \item For every vertex $x$ of $G$, let $e$, $f$, $g$ be the three edges incident with $x$ in $G$.
  Add the edges $w_{x,e}w_{x,f}$, $w_{x,e}w_{x,g}$, $w_{x,f}w_{x,g}$.
\end{enumerate}
\end{sloppypar}
The obtained graph $G'$ has $|V(G)|+5|E(G)|$ vertices and $6|V(G)|+6|E(G)|$ edges, and can be computed in polynomial time from $G$. Moreover,
since a planar embedding of $G$ can be easily transformed into a planar embedding of $G'$, we also have that $G'$ is planar. Since $G$ is $2$-connected, so is $G'$, and $G'$ is of maximum degree $5$.

Furthermore, graph $G'$ is a line graph. To see this, it suffices to observe that $G'$ is isomorphic to the graph obtained from $G$ with the following procedure:
\begin{enumerate}
  \item Subdivide each edge of $G$ twice. Let $G_1$ be the obtained graph.

  \item Add to each vertex $v$ of $G_1$ a private neighbor (that is, a vertex of degree $1$ adjacent to $v$). Let $G_2$ be the obtained graph.

  \item Take the line graph of $G_2$.
\end{enumerate}

To complete the reduction, we need to specify the requirements $r(\cdot)$ to vertices of $G'$. For every edge $e = xy$ of $G$, set $r(w_{x,e}) = r(w_{y,e}) = 4$, and $r(w_{e}) = 3$.
Set $r(v') = 0$ to all other vertices $v'$ of $G'$.

\medskip

\begin{figure}
    \centering
    \includegraphics[width=\linewidth]{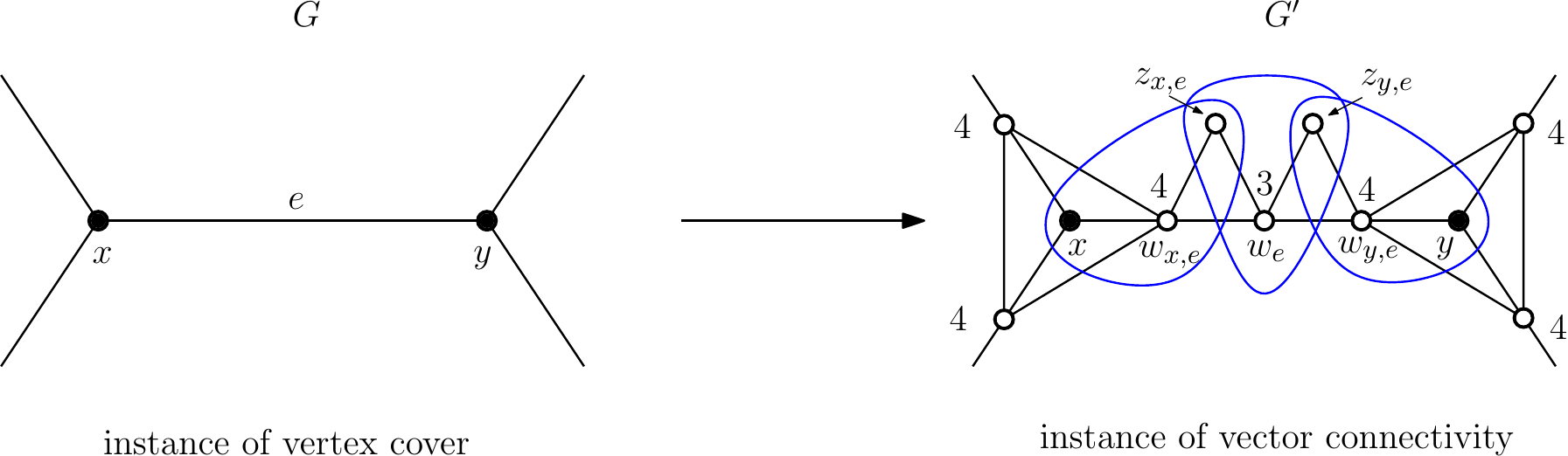}
   \caption{For each edge of $G$, every vector connectivity set for $(G',r)$ must contain at least one vertex from each of the circled sets (see Claim 1);
the three sets can be hit with two vertices, but not with just one. No matter whether one or two of the original (black) vertices are used, one of the additional (white) vertex must be used (and if no black vertex is used, two white vertices must be used).
The best we can do is to find a vertex cover for $G$ and then add one white vertex for each edge of $G$.}
   \label{fig:reduction}
\end{figure}

\begin{sloppypar}
Let $\tau(G)$ denote the minimum size of a vertex cover of $G$. The {\sf NP}-completeness of {\sc VecCon} with maximum requirement $4$ in $2$-connected planar line graphs of maximum degree $5$ will follow from the following lemma.
See Fig.\ \ref{fig:reduction} for a pictorial explanation of the reduction idea.
\end{sloppypar}

\begin{lemma}\label{lem:hard}
$\kappa(G',r) = \tau(G) + |E(G)|$.
\end{lemma}

\medskip
We will prove the lemma through a sequence of auxiliary statements.
Then we will argue how to modify the construction to obtain a bipartite graph of arbitrarily high girth.

\medskip
\noindent
{\it Claim 1: Let $S\subseteq V(G')$. Then, $S$ is a vector connectivity set for $(G',r)$ if and only if for every edge $e = xy$ of $G$,
the set $S$ contains at least one vertex from each of the sets
$X_e :=\{x,w_{x,e},z_{x,e}\}, Z_e:=\{z_{x,e}, w_e,z_{y,e}\}, Y_e:=\{z_{y,e},w_{y,e},y\}$.}

\medskip
{\it Proof of Claim 1:}
Let $e$ be an edge of $G$. Each of the sets $X_e, Y_e$, and $Z_e$ induces a connected subgraph in $G'$. Moreover
$|N_{G'}(X_e)|= |N_{G'}(Z_e)| = 3< 4 = R(X_e) = R(Z_e)$, and $|N_{G'}(Y_e)| = 2< 3 = R(Y_e)$. Thus, if $S$ is a vector connectivity set for $(G',r)$, then Proposition~\ref{prop:char-Menger} implies that $S$ contains a vertex from each of the sets $X_e$, $Y_e$, and $Z_e$.

Conversely, suppose that for every edge $e$ of $G$, set $S$ contains at least one vertex from each of the sets $X_e, Y_e$, and $Z_e$.
Let $v'$ be a vertex in $V(G')\setminus S$. Trivially, if $r(v') = 0$ then $v'$ is $r(v')$-linked to $S$.

Suppose that $r(v') = 3$. Then $v' = w_e$ for some edge $e = xy$ of $G$.
Since $S$ contains a vertex from $Z_e$, we may assume w.l.o.g.~that $z_{x,e} \in S$.
Let $e'= xy'$ be an edge incident with $x$ other than $e$.
Using the assumption on $S$, we can now find $3$ internally vertex disjoint paths linking $v'$ to $S$
with respective endpoints $x_1$, $x_2$, $x_3$ such that $x_1 \in Z_{e'}\cup\{w_{x,e}\}$, $x_2 = z_{x,e}$ and $x_3 \in Y_e$.

Suppose that $r(v') = 4$. Then, up to symmetry, $v' = w_{x,e}$ for some edge $e =xy$ of $G$.
Let $e' = xy'$ and $e'' = xy''$ be the two edges incident with $x$ other than $e$.
Using the assumption on $S$, we can now find $4$ internally vertex disjoint paths linking $v'$ to $S$ with respective
endpoints $x_1,\ldots, x_4$ such that
$x_1 \in X_{e'}\cup Z_{e'}$, $x_2 \in X_{e''}\cup Z_{e''}$,
$x_3 \in \{x,z_{x,e}\}$ and $x_4 \in \{w_e\}\cup Y_e$.

Since every vertex $v'\in V(G')\setminus S$ is $r(v)$-linked to $S$, we conclude that $S$ is a vector connectivity set for $(G',r)$.

\medskip
\noindent
{\it Claim 2:} $\kappa(G',r) \le  \tau(G) + |E(G)|$.

\medskip
{\it Proof of Claim 2:}
Let $C$ be an optimal vertex cover of $G$. First, we extend $C$ to a vertex cover $C'$ of graph $G_1$
(recall that $G_1$ is the graph obtained from $G$ by a double subdivision of each edge).
This can be done by adding exactly one additional vertex for each edge of $G$. Hence, $|C'| \le |C|+|E(G)|$.
Now, consider the set $S\subseteq V(G')$  obtained as follows:
\begin{enumerate}[(1)]
  \item Put in $S$ all vertices of $C$.

  \item For every edge $e = xy$ of $G$, label the vertices of the $4$-vertex path replacing $e$ in $G_1$ as $(x,z_{x,e}',z_{y,e}',y)$.
    Since $z_{x,e}'z_{y,e}'\in E(G_1)$ and $C'$ is a vertex cover of $G_1$, we have $C'$ contains either $z_{x,e}'$ or $z_{y,e}'$.
    If $z_{x,e}'\in C'$, then add $z_{x,e}$ to $S$. Otherwise, add $z_{y,e}$ to $S$.
\end{enumerate}
It can be verified using Claim 1 that   $S$ is a vector connectivity set for $(G',r)$ of total size at most $\tau(G) + |E(G)|$. This
completes the proof of the claim.

\bigskip
For a subset $S$ of $V(G')$, let $n(S)$ denote the total number of elements of the form $w_{x,e}$, $w_e$, $w_{y,e}$ contained in $S$. Equivalently, this is the number of non-simplicial vertices of $G'$ contained in $S$. (A vertex in a graph is said to be {\it simplicial} if its neighborhood forms a clique.
The closed neighborhood of a simplicial vertex $v$ is said to be a {\it simplicial clique (rooted at $v$)}.)

\medskip
\noindent
{\it Claim 3: There is a minimum vector connectivity set for $(G',r)$ with $n(S) = 0$.}

\medskip

{\it Proof of Claim 3:}
Let $S$ be a minimum vector connectivity set for $(G',r)$ that minimizes $n(S)$. Suppose for a contradiction that $n(S)>0$.
By symmetry, we may assume that there exists a vertex $w\in S$ such that $w\in \{w_{x,e},w_{e}\}$ for some edge $e = xy$ of $G$.
By minimality of $S$ and Claim 1, set $S$ does not contain $z_{x,e}$ (otherwise, $S\setminus \{w\}$ would be a vector connectivity set smaller than $S$).
Let $S' = (S\setminus \{w\})\cup\{z_{x,e}\}$. By Claim 1, $S'$ is a vector connectivity set for $(G',r)$. Since $|S'| \le |S|$, set $S'$ is a minimum vector connectivity set. However, we have $n(S')<n(S)$, contradicting the choice of $S$.

\medskip
We are now ready to complete the proof of Lemma~\ref{lem:hard}.
\medskip

{\it Proof of Lemma~\ref{lem:hard}:} Claim 1 established that $\kappa(G',r) \le  \tau(G) + |E(G)|$.

The proof of the reverse inequality can be derived from Claims 2 and 3, as follows. By Claim 3, there exists a minimum vector connectivity set $S$ for $(G',r)$ with $n(S) = 0$. By Claim 2 and since $n(S) = 0$, set $S$ has the following property: for every edge $e=xy$ of $G$, set $S$ contains at least one vertex from each of the sets $\{x,z_{x,e}\}$, $\{z_{x,e},z_{y,e}\}$, $\{z_{y,e},y\}$. Assuming that for every edge $e = xy$ of $G$, the vertices of the $4$-vertex path replacing $e$ in $G_1$ are labeled as $(x,z_{x,e}',z_{y,e}',y)$, the above property implies that the set obtained from
$S$ by replacing each vertex of the form $z_{x,e}$ with $z_{x,e}'$, and
each vertex of the form $z_{y,e}$ with $z_{y,e}'$, is a vertex cover of $G_1$. Therefore, the vertex cover number of $G_1$ is at most $|S| = \kappa(G',r)$.
Since the vertex cover numbers of $G$ and $G_1$ are related with the equality $\tau(G) = \tau(G_1)-|E(G)|$ (see, e.g.,~\cite{MR0351881}), it follows that
$\tau(G) = \tau(G_1)- |E(G)| \le \kappa(G',r)- |E(G)|$, as desired.

\bigskip

It remains to show how to modify the construction to obtain a bipartite graph of girth at least $k$
(while keeping planarity, $2$-connectivity and maximum degree).
We take the graph $G'$ and subdivide each edge $2k+1$ times. Since we subdivided each edge an odd number of times, the obtained graph $\tilde G$
is bipartite. Clearly, girth of $\tilde G$ is more than $k$, while planarity, $2$-connectivity and maximum degree are maintained.
On original vertices, we keep the requirements as above. To each new vertex, we assign requirement $0$. Let $\tilde r:V(\tilde G)\to \{0,3,4\}$ be the new requirement function.
In this setting, equality $\kappa(\tilde G, \tilde r) = \tau(G) + |E(G)|$ holds, which will establish {\sf NP}-completeness.

The inequality $\le$ can be proved similarly as the inequality  $\kappa(G', r) \le \tau(G) + |E(G)|$ above in Claim 1.
For the proof of the reverse inequality, we need to modify Claims 1 and 3 appropriately. First, observe that each edge of $G'$ is contained in a unique simplicial clique (rooted at a vertex of the form $x$ or $z_{x,e}$ where $x\in V(G)$ and $e\in E(G)$). Hence, we can associate to each edge $e'$ of $G'$ the simplicial vertex $s(e')$ of $G'$ such that $e'\subseteq N_{G'}[s(e')]$. For a simplicial vertex $s$ of $G'$, let $A(s) = \{s\}\cup B(s)$, where $B(s)$ is the set of all vertices in $\tilde G$ that were introduced on some edge $e'$ of $G'$ with $s(e') = s$ in the process of constructing $\tilde G$ from $G'$.

In this setting, Claim 1 is replaced with the following claim.

\medskip
\noindent
{\it Claim 1': Let $\tilde S$ be a subset of $V(\tilde G)$. Then, $\tilde S$ is a vector connectivity set for $(\tilde G, \tilde r)$ if and only if
 for every edge $e =xy$ of $G$, set
$\tilde S$ contains at least one vertex from each of the sets
$A(x) \cup \{w_{x,e}\} \cup A(z_{x,e})$, $A(z_{x,e}) \cup \{w_e\} \cup A(z_{y,e})$, $A(z_{y,e}) \cup \{w_{y,e}\} \cup A(y)$.}

\medskip

The proof is similar to that of Claim 1.

For a subset $\tilde S$ of $V(\tilde G)$, let $n(\tilde S)$ denote the total number of non-simplicial vertices of $\tilde G$ contained in $\tilde S$.
Claim 3 is replaced with the following claim.

\medskip
\noindent
{\it Claim 3': There is a minimum vector connectivity set for $(\tilde G, \tilde r)$ with $n(\tilde S) = 0$.}

\medskip

The proof is similar to that of Claim 3. Vertices of the form $x\in V(G)$ and $w_{e}$ for $e\in E(G)$ are handled similarly as in the proof of Claim 2.
If $\tilde S$ contains a non-simplicial vertex $w$ of degree $2$, we can simply replace it with vertex $s$, where $s$ is the simplicial vertex of $\tilde G$ such that $w\in A(s)$.
\end{proof}

A similar approach as one used to prove Theorem~\ref{thm:NP-c} can be used to show the following inapproximability result.
Recall that {\sf APX} is the class of problems approximable in polynomial time to within some constant,
and that a problem $\Pi$ is {\it {\sf APX}-hard} if every problem
in {\sf APX} reduces to $\Pi$ via an AP-reduction~\cite{MR1734026}.
{\sf APX}-hard problems do not admit a polynomial-time approximation scheme (PTAS), unless {\sf P} = {\sf NP}.
To show that a problem is {\sf APX}-hard, it suffices to show that an {\sf APX}-complete problem
is L-reducible to it~\cite{MR1734026}.
Given two {\sf NP} optimization problems $\Pi$ and $\Pi'$, we say that
$\Pi$ is {\it $L$-reducible} to $\Pi'$ if
there exists a polynomial-time transformation
$f$ from instances of $\Pi$ to instances of $\Pi'$
and positive constants $\alpha$ and $\beta$ such that for every instance $x$ of $\Pi$, we have:
\begin{enumerate}
  \item $\text{opt}_{\Pi'}(f(x))\le \alpha \cdot \text{opt}_{\Pi}(x)$, and
  \item for every feasible solution $y'$ of $f(x)$ with objective value
  $c_2$ we
can compute in polynomial time a solution $y$ of $x$ with objective value
$c_1$ such that
  $|\text{opt}_{\Pi}(x)-c_1|\le \beta \cdot |\text{opt}_{\Pi'}(f(x))-c_2|$.
\end{enumerate}

\begin{theorem}\label{thm:APX}
{\sc VecCon} is {\sf APX}-hard. In particular,
{\sc VecCon} admits no PTAS, unless {\sf P = NP}.
\end{theorem}

\begin{proof}
Since {\sc Vertex Cover} is APX-complete for cubic graphs~\cite{MR1756204}, it suffices to show that
{\sc Vertex Cover} in cubic graphs is $L$-reducible to {\sc VecCon}.
Consider the polynomial-time transformation described in the first part of the proof of Theorem~\ref{thm:NP-c},
that starts from a cubic graph $G$ (not necessarily planar or $2$-connected), an instance to
{\sc Vertex Cover}, and computes an instance $(G',r)$ of {\sc VecCon}.

By Lemma~\ref{lem:hard}, we have $\kappa(G',r) = \tau(G) + |E(G)|$.
Moreover, since $G$ is cubic, every vertex in a vertex cover of $G$
covers exactly $3$ edges, hence $\tau(G)\geq \frac{|E(G)|}{3}$.
This implies that $\kappa(G',r) = \tau(G) + |E(G)|\le 4\tau(G)$, hence the first condition
in the definition of $L$-reducibility is satisfied with
$\alpha = 4$.
The second condition in the definition of $L$-reducibility
states that for every vector connectivity set $S$ of $(G',r)$ we can
can compute in polynomial time a vertex cover $C$ of $G$
such that
$|C|-\tau(G)\le \beta \cdot |S|-\kappa(G,r)$ for some $\beta>0$.
We claim that this can be achieved with $\beta = 1$.
Indeed, the proof of Lemma~\ref{lem:hard} above shows how one can transform in polynomial time any vector connectivity set $S$
in $G'$ to a vertex cover $C$ of $G$ such that $|C|\le |S|-|E(G)|$.
Therefore,
$|C|-\tau(G)\le |S|-|E(G)|-\tau(G) = |S|-\kappa(G',r)$.
This shows that {\sc Vertex Cover} in cubic graphs is $L$-reducible to
{\sc VecCon}, and completes the proof.
\end{proof}


\section{Reduction to biconnected graphs}\label{sec:reduction}

A {\it cut vertex} in a connected graph $G$ is a vertex $v$ such that $G-v$ is disconnected.
A graph is {\it biconnected} if it is connected and has no cut vertices.
A {\it block} (or a {\it biconnected component}) of a connected graph $G$ is a maximal biconnected subgraph of $G$.
The blocks of a connected graph $G$ are connected in a tree structure, which is called the {\it block tree} of $G$, and whose leaves are
referred to as the {\em leaf blocks} of $G$. The block tree of a given connected graph can be computed in linear time.

Boros et al.~\cite{BHHM} proved that {\sc VecCon} is polynomial for trees.
In this section, we present an algorithm showing that
{\sc VecCon} is polynomial also for the larger class of
{\it block graphs}, that is, graphs every block of which is complete.
Our result will follow from a solution to a more general problem on block graphs,
which we call {\sc Free-Set Vector Connectivity} ({\sc FreeVecCon}). In this problem, a subset $F$ of {\it free vertices}
is also given as part of the input, and the requirement that $S$ is a vector connectivity set is relaxed to the requirement that every vertex
in $V\setminus S$ is $r(v)$-linked to $S\cup F$. When $F = \emptyset$, the problem is equivalent to {\sc VecCon}.
Formally, given a graph $G = (V,E)$, an integer-valued function $r:V\to \mathbb{Z}_+$ and a subset $F\subseteq V$,
a {\it vector connectivity set} for $(G,F,r)$ is a set $S\subseteq V$ such that every vertex in
$V\setminus S$ is $r(v)$-linked to $S\cup F$.
The {\sc FreeVecCon} problem is the problem of finding a vector connectivity set of minimum size for $(G,F,r)$.
Given a graph $G$, a vertex $v\in V(G)$ and a subset $S\subseteq V(G)$,
we will denote by $\kappa_G(v,S)$ the maximum $k$ such that $v$ is $k$-linked to $S$.

We present an algorithm with which we can show that if for a given hereditary class of
graphs ${\cal G}$, the {\sc FreeVecCon} problem can be solved in polynomial time on biconnected graphs in ${\cal G}$ then
the {\sc FreeVecCon} problem is also  solvable in polynomial time on graphs from $\cal G.$
Clearly, when solving the {\sc VecCon} and {\sc FreeVecCon} problems we may restrict our attention to connected graphs.

\begin{algorithm}[h!]
\caption{Reduces {\sc FreeVecCon} to the biconnected case}
\label{algo2:main}
\small
{\bf Procedure} {\tt FSVecCon}($G,F,r$)\\
{\bf Input:} A connected graph $G=(V,E)$, a subset $F\subseteq V$, a function $r:V\to \mathbb{Z}_+$.\newline
{\bf Output:} A minimum vector connectivity set for $(G,F,r)$.\\
\vspace{-0.3cm}
\begin{algorithmic}[1]
\IF {$G$ is biconnected }
\STATE{{\bf return} {\tt FSVecCon-biconnect}($G,F,r$)}
\ENDIF
\medskip
\STATE{{\bf if} $\emptyset$ is a vector connectivity set for $(G,F,r)$ {\bf then return} $\emptyset$}
\medskip
\STATE{Let $B$ be a leaf block of $G$; \; $v$ the cut vertex contained in $B$;  \; $R \leftarrow G - (V(B) \setminus \{v\})$;}
\STATE{Let $\beta$ and $\rho$ be two Boolean values computed as follows:}
\STATE{~~~$\beta = {\tt true}$ {\bf iff} every vertex $u\in V(B)\setminus \{v\}$ is $r(u)$-linked to $F\cup \{v\}$}
\STATE{~~~$\rho = {\tt true}$ {\bf iff} every vertex $w\in V(R)\setminus \{v\}$ is $r(w)$-linked to $F\cup \{v\}$;}
\medskip
\STATE{{\bf if} $\beta \wedge \rho = {\tt true}$ {\bf then return} $\{v\}$}
\medskip
\IF {$\beta \wedge \rho = {\tt false}$ {\bf and} $\beta \vee \rho = {\tt true}$ (that is, exactly one between $\beta$ and $\rho$  is {\tt true})}
\STATE {{\bf if} $\beta = {\tt true}$ {\bf then} $H \leftarrow R$ {\bf else } $H \leftarrow B$}
\IF{$F \not \subseteq V(H) \setminus \{v\}$}
\STATE{$r(v) \leftarrow \max\{r(v) - \kappa_G(v, F \setminus V(H)) + |\{v\} \setminus F|,0\}$}
\STATE{$F \leftarrow F \cup \{v\}$}
\ENDIF
\STATE{{\bf return} {\tt FSVecCon}($H,F\cap V(H), r|_{V(H)}$)}
\ENDIF
\medskip
\IF {$\beta \vee \rho = {\tt false}$} \label{algo2:block3-s}
\FOR {$i \in \{0,1,\ldots,r(v)\}$} \label{algo2:for-s}
 \STATE{$r(v) \leftarrow i$; \;\;  $S_i \leftarrow$ {\tt FSVecCon-biconnect}$(B, (F \cap V(B)) \cup \{v\}, r|_{V(B)})$}
\ENDFOR \label{algo2:for-e}
\STATE{$i^* = \max\{j \mid 0\le j\le r(v),\, |S_j| = |S_0|\}$}
\STATE{$r(v) \leftarrow r(v) - i^* + 2$}
\STATE{$S_R \leftarrow {\tt FSVecCon}(R, (F \cap V(R))\cup \{v\}, r|_{V(R)})$}
\STATE{{\bf return} $S_{i^*} \cup S_R$}
\ENDIF \label{algo2:block3-e}
\end{algorithmic}
\end{algorithm}

\begin{theorem}\label{thm:reduction-general}
Suppose that there exists an ${\mathcal O}(T)$ algorithm for
the {\sc FreeVecCon} problem on
biconnected graphs from a class ${\cal G}$. Then,
the {\sc FreeVecCon} problem on a connected graph $G$ every block of which is in ${\cal G}$
can be solved in time
\begin{equation} \label{eq:time-bound}
{\mathcal O}\left(|V(G)|^2\min\left\{r_{\rm max}, |V(G)|^{1/2}\right\} \min\left\{ r_{\max}|V(G)|, |E(G)|\right\} + |V(G)|r_{\rm max}T\right)\,,
\end{equation}
where
$r_{\rm max}=\max\{2,\max\{r(v)\mid v\in V(G)\}\}$.
\end{theorem}

\begin{sloppypar}
\begin{proof}
Assume that {\tt FSVecCon-biconnect} is an algorithm that in polynomial time ${\mathcal O}(T)$
correctly solves the {\sc Free Vector Connectivity} problem on instances $(G,F,r)$ such that $G\in{\cal G}$ is biconnected and
$r:V(G)\to\{0,1,\dots,r_{\max}\}$. We claim that {\bf Procedure {\tt FSVecCon}}
correctly solves {\sc FreeVecCon} in time given by (\ref{eq:time-bound}).

We will first prove the correctness and optimality of the solution returned by {\tt FSVecCon}. At the end, we will analyze the time complexity of the algorithm.

\medskip

\noindent
{\bf Lines 1--3.}
If $G$ is biconnected, then in lines 1--3, the solution is constructed using procedure
{\tt FSVecCon-biconnect}. In this case, feasibility and optimality follow from the assumption on
{\tt FSVecCon-biconnect} which is a polynomial time algorithm  for {\sc FreeVecCon}
restricted to instances $(G,F,r)$ such that $G$ is biconnected and $r:V(G)\to\mathbb{Z}$.

\medskip

If $G$ is not biconnected, then
in line 4, the algorithm checks if the trivial solution works, and if not,
then in lines 5--8, the algorithm identifies a leaf block $B$ and
computes the Boolean values $\beta$ and $\rho.$ These values are meant to indicate whether
the set of free vertices together with the cut vertex satisfy the connectivity requirements in the
leaf block ($\beta$) and in the  remaining part of the graph ($\rho$), respectively.

\medskip
More precisely, let $B$ be a leaf block of $G$ and let $v\in V(B)$ be the cut vertex of $G$ contained in $B$.
Let $R = G-(V(B)\setminus\{v\})$. The algorithm computes the values
$$\beta = \left\{
               \begin{array}{ll}
                 \texttt{true}, & \hbox{if every vertex $u\in V(B)\setminus \{v\}$ is $r(u)$-linked to $F\cup \{v\}$;} \\
                 \texttt{false}, & \hbox{otherwise.}
               \end{array}
             \right.$$
and $$\rho = \left\{
               \begin{array}{ll}
                 \texttt{true}, & \hbox{if every vertex $w\in V(R)\setminus \{v\}$ is $r(w)$-linked to $F\cup \{v\}$;} \\
                 \texttt{false}, & \hbox{otherwise.}
               \end{array}
             \right.$$

\medskip

The remaining part of the computation is based on the values $\beta$ and $\rho$ and may imply recursive calls of the algorithm.
In the analysis below, whenever the meaning will be clear from the context, we will not distinguish between a graph and its vertex set.
Let ${OPT}$ denote the minimum size of a vector connectivity set for $(G,F,r)$.

\medskip

\noindent
{\bf Case A.}
$\beta = \rho = \texttt{true}$.

This together with the fact that the empty set is not a solution implies that ${OPT} = 1.$ In fact,
the assumption that $\beta = \rho = \texttt{true}$ implies that $\{v\}$ is a vector connectivity set for $(G,F,r).$

\medskip

\noindent
{\bf Case B.} Exactly one of the values $\beta$ and $\rho$ is true.

This means that the cut vertex $v$ and $F$
can satisfy the requirements of only one of the two parts of the graph on the two sides of the cut vertex. The algorithm is then
recursively invoked on the part of the graph in which the requirements are not satisfied. In the following we show that the solution so
computed is indeed a solution for the original problem. Since the two cases are similar to each other, we shall limit ourselves to
discuss only the case where $\beta = {\tt true}$ and $\rho = {\tt false}.$
Under this assumption, we have $H = R$, and the condition
$F \not \subseteq V(H) \setminus \{v\}$ in line 12 is equivalent to the condition
$F \cap B\neq\emptyset$.

\noindent
{\em Observation 1.} Since $\rho = \texttt{false}$, there exists a vertex
$w\in V(R)\setminus\{v\}$ that is not $r(w)$-linked to $F\cup\{v\}$.
In particular, $w$ is not $r(w)$-linked to $F$, therefore any solution $S$ to
the {\sc FreeVecCon} problem on $(R, F \cap R, r'),$ with $r'(w) = r(w)$ for each $w \neq v,$ is non-empty.
The same holds true for the solutions to the problem on $(R, (F \cap R) \cup \{v\}, r'),$

\smallskip

\noindent
{\em Observation 2.}
There exists a minimum vector connectivity set $S^*$ for
$(G,F,r)$ such that $S^*\subseteq R$.  Indeed, if $S^*$ is a minimum vector connectivity set for $(G,F,r)$ and $u\in S^*\setminus R$, then
$(S^*\setminus\{u\})\cup\{v\}$ is a vector connectivity set for $(G,F,r)$.

\medskip
We shall split the analysis into two cases according to the intersection of $F \cap B$.
\medskip

\noindent
{\em Subcase B-1.} $F\cap B = \emptyset$.

Then the algorithm is called recursively on $(R, F, r|_{R})$ and the solution obtained is returned.

\medskip

{\it Feasibility.}
Since $\beta = \texttt{true}$, every vertex $u\in B\setminus\{v\}$ is $r(u)$-linked to $F\cup \{v\}$.
Since $F\cap B=\emptyset$ and $v$ is a cut vertex, every vertex $u\in B\setminus\{v\}$
is $r(u)$-linked to $\{v\}$, hence $r(u)\le 1$.
On the other hand,  if $S$ is the solution returned, which is non-empty by Observation 1,  every vertex $u\in B\setminus\{v\}$ is
$r(u)$-linked to $S$ (and hence to $S\cup F$). Since $S$ is a vector connectivity set for  $(R, F, r|_{R})$,
we have that $S$ is also a vector connectivity set for $(G,F,r)$, which establishes feasibility.

{\it Optimality.}
Let $S^*$ be a minimum vector connectivity set for $(G,F,r)$ such that $S^*\subseteq R$ (such a set exists by Observation 2).
But now, every vertex $w\in R\setminus S^*$ is $r(w)$-linked to $S^*\cup F$.
Since $S^*\cup F\subseteq R$, every path connecting $w$ to $S^*\cup F$ lies entirely in $R$.
Thus, $S^*$ is a vector connectivity set for  $(R, F, r|_{R})$. By the choice of $S$, we have $|S|\le |S^*|$, which establishes optimality.

\medskip

\noindent
{\it Subcase B-2.} $F \cap B \neq \emptyset.$
In this case, we define a new requirement assignment for the vertices in $R$ as follows.
Let $r':V(R)\to\{0,1,\dots,r_{\max}\}$ be defined as
$$r'(w) = \left\{
                                                     \begin{array}{ll}
                                                       \max\{ r(v)-\kappa_G(v,F\cap(B\setminus\{v\})) + |\{v\} \setminus F|,0\} & \hbox{if $w = v$;} \\
                                                       r(w), & \hbox{otherwise.}
                                                     \end{array}
                                                   \right.$$
for all $w\in V(R)$.
Then, the algorithm is recursively executed on $(R, (F\cap R) \cup \{v\}, r')$ and the obtained solution is returned.

\medskip

{\it Feasibility.}
Let $S$ be the solution returned in this step.
By definition, and Observation 1 we have $\emptyset \neq S\subseteq R$, and every vertex $w\in R\setminus S$ is $r'(w)$-linked to $(S\cup (F\cap R))\cup \{v\}$ (in $R$).
To establish feasibility, we need to verify that every vertex $w\in G\setminus S$ is $r(w)$-linked to $S\cup F$ (in $G$).
If $w\in R\setminus \{v\}$, then the condition holds due to above property of $S$ and the fact that $r(w) = r'(w)$.
If $w\in B\setminus\{v\}$, by  the assumption that $\beta = \texttt{true}$, it follows
that $w$ is $r(w)$-linked to $F\cup \{v\}$. Now, if  $v\in F$, it immediately follows that
$w$ is $r(w)$-linked to $S\cup F$. On the other hand, if $v \not \in F,$ it is easy to see that $w$ is $r(w)$-linked to $S \cup F$ because
of $\emptyset \neq S\subseteq R$ since we can extend the path to $v$ to a path to $S.$

Finally, suppose that $w = v$.
Let $k = \kappa_G(v,F\cap(B\setminus\{v\}))$.
Vertex $v$ is $k$-linked to $F\cap (B\setminus\{v\})$ in $G$ (equivalently: in $B$).
If $k\ge r(v)$ then $v$ is also $r(v)$-linked (in $G$) to $S\cup F$.
If $k < r(v)$, then $r(v)-k + |\{v\} \setminus F|\ge 0$, and
$v$ is also $(r(v)-k+|\{v\}\setminus F|)$-linked in $R$ to
$(S\cup (F \cap R)) \cup \{v\}$. This immediately implies that:
(i) if $v \in F$ then
$v$ is $r(v)$-linked (in $G$) to $S \cup F$;
if $v \not \in F$ then $v$ is  $(r(v)+1)$-linked to
$S \cup F \cup \{v\},$ hence it is  $r(v)$-linked (in $G$) to $S \cup F.$

\medskip

{\it Optimality.}
Let $S^*$ be a minimum vector connectivity set for $(G,F,r)$ such that $S^* \subseteq R$ (such a set exists by Observation 2).
It suffices to show that $S^*$ is a vector connectivity set for
$(R, (F\cap R)\cup\{v\}, r')$, since then we will have $|S|\le |S^*|$, establishing optimality.
Let $w\in R\setminus S^*$.
Since $S^*$ is a vector connectivity set for $(G,F,r)$, vertex $w$ is \hbox{$r(w)$-linked} to $S^*\cup F$ in $G$.
If $w\neq v$, then $r(w) = r'(w)$, hence $w$ is $r'(w)$-linked to $S^*\cup F$ in $G$.
Since $v$ is a cut vertex, this implies that
$w$ is $r'(w)$-linked to $S^*\cup (F\cap R)\cup \{v\}$ in $R$.
Suppose now that $w = v$. If $r'(v) = 0$, then $v$ is trivially $r'(v)$-linked in to $S^*\cup (F\cap R)\cup \{v\}$ in $R$.
If $r'(v) > 0$, then $r'(v) = r(v)-\kappa_G(v,F\cap(B\setminus\{v\}))+|\{v\}\setminus F|$,
and $v$ is $r(v)$-linked to $S^*\cup F$ in $G$.
Therefore
\begin{eqnarray*}
&&\kappa_G(v,F\cap(B\setminus\{v\})) + \kappa_R(v,S^*\cup (F\cap R)\cup \{v\})\\&=&
\kappa_G(v,S^*\cup F\cup \{v\})\\
&=&
\kappa_G(v,S^*\cup F) + |\{v\}\setminus F|  \\
&\ge& r(v)+|\{v\}\setminus F|\\
&=& \kappa_G(v,F\cap(B\setminus\{v\}))+r'(v)\,,
\end{eqnarray*}
and $v$ is again $r'(v)$-linked to $S^*\cup (F\cap R)\cup \{v\}$ in $R$.

\bigskip

\noindent
{\bf Case C.} $\beta = \rho = {\tt false}.$

The idea is to first compute a solution $S_B$ for the block $B$ that does not include
vertex $v$ and which maximizes the value $k$ such that $v$ is $k$-linked to $S_B \cup (F \cap B).$  To this aim, we compute the
sets $S_i,$ for $i=0,\ldots, r(v),$ defined as the solution to the instance $(B, (F\cap B)\cup\{v\}, r_i)$, where $r_i$ is the
restriction of $r$ to $B$, modified so that the requirement of $v$ is $r_i(v) = i.$

In formulae, for all $i \in \{0,1,\ldots, r(v)\}$,
let $r_i$ be a new requirement assignment defined as
$$r_i(u) = \left\{
                                                     \begin{array}{ll}
                                                       i, & \hbox{if $u = v$;} \\
                                                       r(u), & \hbox{otherwise.}
                                                     \end{array}
                                                   \right.$$
Using algorithm {\tt FSVecCon-biconnect}, compute an optimal solution $S_i$
 to the {\sc FreeVecCon} problem on $(B, (F\cap V(B))\cup \{v\}, r_i)$.
Let $i^*$ be the maximum in $\{0, 1, \ldots, r(v)\}$ such that $|S_{i^*}| = |S_0|.$

\medskip
\noindent
{\it Observation 3.} $v\not\in S_{i^*}$. Indeed, the set $S = S_{i^*}\setminus \{v\}$ is a vector connectivity set for $(B,(F\cap B)\cup\{v\},r_0)$.
Thus $|S_0| \le |S| \le |S_{i^*}|= |S_0|$, which implies that $|S| = |S_{i^*}|$ and consequently $v\not\in S_{i^*}$.
In particular, the set $T = S_{i^*} \cup (F \cap B) \cup \{v\}$ satisfies all the requirements in $B \setminus \{v\}$ and
since $v\not\in S_{i^*}$, vertex $v$ is $i^*$-linked to $T.$

%
\medskip
We set  $S_B = S_{i^*}.$
\medskip

{\it Feasibility.}
First, let us argue that we have $r(v)-i^*+2\le r_{\max}$,\label{thispage} thus the instance
$(R, (F \cap R)\cup \{v\}, r|_{R})$ from the recursive call in line 24 has requirements at most $r_{\max}$.
Indeed, suppose for a contradiction that $r(v)-i^*+2\ge r_{\max}+1$.  This implies that $r(v)\ge i^*+r_{\max}+1$.
Since $r(v)\ge 1,$ by the definition of $i^*$ it also follows that $i^*\ge 1$. Hence, it must be $r(v) = r_{\max}$ and $i^*= 1$.
By definition of $i^*$, vertex $v$ is not $2$-linked to $(S_0\cup (F\cap B))\cup \{v\}$, in particular,
$S_0 = \emptyset$.
Consequently, every vertex $u\in B\setminus \{v\}$ is $r(u)$-linked to $S_0\cup (F \cap B) \cup \{v\} \subseteq F\cup \{v\}$, which is
a contradiction to the assumption that $\beta = {\tt false}$.

Since $\beta$ and $\rho$ are both {\tt false}, it follows that $|S_B \cap (B \setminus \{v\})| \geq 1$ and
$|S_R \cap (R \setminus \{v\})| \geq 1.$ Therefore, from $S_R$ (resp.~$S_B$) there is a path to $v$.
Each $w \in (B \setminus S_B) \setminus \{v\}$ (resp.~$w \in (R \setminus S_R) \setminus \{v\}$) is $r(w)$-linked to $S_B \cup (F \cap B) \cup \{v\}$
(resp.~$S_R \cup (F \cap R) \cup \{v\}$). Therefore, if one of such paths is to $v$, it can be extended to a path to
$S_R$ (reps.~$S_B$). Hence $w$ is $r(w)$-linked to $S_B \cup S_R \cup F$, as required.

As for $v,$ by the definition of $S_R$ and $S_B$ we have that $v$ is at least $(i^*-1)$-linked to $S_B$ and
at least $(r(v) - i^* +1)$-linked to $S_R \cup (F \cap (R \setminus \{v\})).$ This follows from having required of $S_B$ and $S_R$ that
$v$ is $i^*$-linked to $S_B \cup (F \cap B) \cup \{v\},$ and $v$ is $(r(v)-i^*+2)$-linked to $S_R \cup (F \cap R) \cup \{v\}.$
Altogether we have $\kappa_G(v, S_B \cup S_R \cup F) \geq r(v),$ as required.

\medskip

{\it Optimality.}
Recall that $v \not \in S_B = S_{i^*}$ (by Observation 3).
 Let $S^*$ be an optimal solution.
Because of $\rho = {\tt false},$ we have $|S^* \cap (R \setminus \{v\})| \geq 1.$
Analogously, $\beta = {\tt false}$ implies $|S^* \cap (B \setminus \{v\})| \geq 1.$
Furthermore, since $S^* \cap (B \setminus \{v\})$ is a vector connectivity set for
$(B, (F \cap B) \cup \{v\}, r_0)$, we have $|S^* \cap (B \setminus \{v\})|\ge |S_0| = |S_B|$.

We now argue according to several cases, depending on whether $v\in S^*$ or not and depending on the value of $i^*$.

\noindent
{\em Case 1.} $v\in S^*$.

In this case, the set $S^* \cap R$ is a feasible solution to the instance
$(R, (F \cap R) \cup \{v\}, r'),$ where
 $r'(w) = r(w)$ for each $w \in R \setminus \{v\}$ and $r'(v) = r(v) - i^*+2.$
Hence, $|S^* \cap R|\ge |S_R|$ and we have
$$|S^*| \ge |S^* \cap (B\setminus \{v\})| + |S^* \cap R| \geq |S_B| + |S_R| = |S_B \cup S_R|$$
where the last equality follows by $S_B \cap S_R \subseteq \{v\} \not \subseteq S_B$ (by Observation 3).

\medskip
\noindent
{\em Case 2.} $v\not \in S^*$.

\medskip
\noindent
{\em Subcase 2.1.} $i^* = r(v)$. In this case we have   $|S_0| = |S_1| = \ldots =  |S_{r(v)}|.$

We have $S_B = S_{r(v)}$ and $S_R$ is the optimal solution to the instance
$(R, (F \cap R) \cup \{v\}, r'),$ where $r'(w) = r(w)$ for each $w \neq v$ and $r(v) = 2.$
Note that the set $S^* \cap R$ is a feasible solution for the instance $(R, (F \cap R) \cup \{v\}, r').$ This easily follows from
\hbox{$S^* \cap R = S^* \cap (R \setminus \{v\}) \neq \emptyset$}, which guarantees the fact that $v$ is at least $2$-linked to
$S^* \cup ((F \cap R) \cup \{v\})$.
Therefore, we have $|S^* \cap R| \geq |S_R|.$
Altogether we get
$$|S^*| \geq  |S^* \cap (B\setminus \{v\})| + |S^* \cap R| \geq  |S_B| + |S_R| = |S_B \cup S_R|.$$

\medskip
\noindent
{\em Subcase 2.2.} $i^* < r(v)$. In this case we have $|S_0| = |S_1| = \ldots = |S_{i^*}| < |S_{i^*+1}|$.

Let $k = \kappa_G(v, (S^* \cup F \cup \{v\}) \cap B)$.

Suppose first that $k\ge i^*+1$. Then the set $S^*\cap B$ is a feasible solution for the instance
$(B, (F \cap B) \cup \{v\}, r_{i^*+1})$, and hence $|S^*\cap B|\ge |S_{i^*+1}|\ge |S_{i^*}|+1 = |S_{B}|+1$.
Moreover,
the set $(S^* \cap R)\cup\{v\}$ is a solution to the instance
$(R, (F \cap R) \cup \{v\}, r'),$ where
 $r'(w) = r(w)$ for each $w \in R \setminus \{v\}$ and $r'(v) = r(v) - i^*+2.$
Hence, $|(S^* \cap R) \cup \{v\})|\ge |S_R|$, that is,
$|S^* \cap R|\ge |S_R|-1$.
Altogether we get
$$|S^*| \geq  |S^* \cap B| + |S^* \cap R| \geq  |S_{B}|+1  + |S_R|-1 = |S_B \cup S_R|.$$

Suppose now that $k\le i^*$. In this case,
since $$k+\kappa_R(v,(S^*\cup F)\cap R) = \kappa_G(v, S^* \cup F \cup \{v\}) \ge r(v)+1\,,$$ we get
$\kappa_R(v,(S^*\cup F)\cap R)\ge r(v)-i^*+1$.
Hence, $v$ is $(r(v)-i^*+2)$-linked to \hbox{$(S^*\cap R)\cup (F \cap R) \cup \{v\}$},
and the set $S^* \cap R$ is a solution to the instance
$(R, (F \cap R) \cup \{v\}, r'),$ where  $r'(w) = r(w)$ for each $w \in R \setminus \{v\}$ and $r'(v) = r(v) - i^*+2.$
Hence,
$|S^* \cap R|\ge |S_R|$.
Recall that $|S^* \cap (B\setminus \{v\})|\ge |S_B|$, and so altogether we get
$$|S^*| \geq  |S^* \cap (B\setminus \{v\})| + |S^* \cap R| \geq  |S_B| + |S_R| = |S_B \cup S_R|.$$

 This establishes the correctness of the algorithm.
It remains to analyze its time complexity.

We need a single call to the procedure
{\tt FSVecCon-biconnect}. Given a graph $H$, a set $U\subseteq V(H)$, a vertex $v\in V(H)$ and an integer $k$,
it can be checked whether vertex $v$ is $k$-linked to $U$ in time
\begin{equation} \label{eq:k-linkedness-time}
{\mathcal O}\left(\min\left\{k^2|V(H)|, k|V(H)|^{3/2}, k|E(H)|, |V(H)|^{1/2}|E(H)| \right\}\right)
\end{equation}
The first two terms in the minimum come from the result of
Nagamochi and Ibaraki~\cite{MR1154589}. The last two terms come from the result of
Even and Tarjan~\cite{EvenTarjan}.

\begin{sloppypar}
Therefore, using the fact that $r(v)\le r_{\max}$ for all $v\in V$, the Boolean values $\beta$ and $\rho$ in lines 6--8
can be computed in time
$${\mathcal O}(|V(G)| \min\{r_{\max},  |V(G)|^{1/2}\} \min\{r_{\max}|V(G)|, |E(G)|\}).$$
The same time complexity can be similarly shown for line 9.
For lines 10--17, observe that lines 10 and 11 take ${\mathcal O}(1)$ time and line 12 takes ${\mathcal O}(|V(G)|)$ time.
Line 13 can also be implemented in
$${\mathcal O}(r_{\max}  \min\{r_{\max},  |V(G)|^{1/2}\} \min\{r_{\max}|V(G)|, |E(G)|\})$$ time since it suffices to check whether
$\kappa_G(v,F\setminus V(H))\ge i$ for each
\hbox{$i\in \{1,\ldots, r(v)+|\{v\}\setminus F|\}\subseteq \{1,\ldots, r_{\max}+1\}$}, and
we can use the bound in (\ref{eq:k-linkedness-time}).
In line 16, one recursive call to the procedure is made.
Lines 18--26 require one recursive call to the procedure, while the remaining steps in these lines can be implemented in time ${\mathcal O}(|V(G)|)$.
\end{sloppypar}

In either case, at most one recursive call to the algorithm is made, and the rest takes
${\mathcal O}(|V(G)| \min\{r_{\max},  |V(G)|^{1/2}\} \min\{r_{\max}|V(G)|, |E(G)|\})$ time.
Since there are at most ${\mathcal O}(|V(G)|)$ recursive calls,
this results in a total complexity of ${\mathcal O}(|V(G)|^2 \min\{r_{\max},  |V(G)|^{1/2}\} \min\{r_{\max}|V(G)|, |E(G)|\})$,
together with at most ${\mathcal O}(r_{\max}|V(G)|)$ calls of the procedure {\tt FSVecCon-biconnect}
(at most ${\mathcal O}(r_{\max})$ calls for each block of $G$).
Thus, the claimed time complexity follows.
\end{proof}
\end{sloppypar}

\section{Polynomiality in block graphs}\label{sec:poly}

Using Theorem~\ref{thm:reduction-general}, we can show that the {\sc FreeVecCon} problem is polynomial in the
class of block graphs, thus generalizing in two ways the polynomial-time solvability of {\sc VecCon} on trees due to Boros et al.~\cite{BHHM}.
To obtain this result, it suffices to argue that {\sc FreeVecCon} is polynomial on complete graphs.
Given a complete graph $K=(V,E)$ with vertex requirements $r:V\to\mathbb{Z}_+$ and a free set $F\subseteq V$,
a simple exchange argument applied separately to $F$ and $V-F$ implies that
there exists an optimal solution consisting of $k$ largest requirement vertices from $F$ and $\ell$ largest requirement vertices
from $V\setminus F$, for some $k\in \{0,\ldots, |F|\}$ and some $\ell\in \{0,\ldots, |V|-|F|\}$.
An optimal solution can thus be found by sorting the vertices in $F$ and $V\setminus F$ according to their requirements, and
checking which of the $O(|V(K)|^2)$ pairs $(k,\ell)$ as above minimizes the value of $k+\ell$ subject to the constraint that $k$ largest requirement vertices from $F$ and $\ell$ largest requirement vertices from $V-F$ form a vector connectivity set for $(K,F,r)$.
Therefore, Theorem~\ref{thm:reduction-general} yields the following.

\begin{theorem}
{\sc VecCon} and {\sc FreeVecCon} problems are solvable in polynomial time on block graphs.
\end{theorem}

A similar approach can be used to show that, more generally, {\sc FreeVecCon} is polynomial also for
the class of {\it block-cactus graphs}, that is, graphs every block of which is either a complete graph or a cycle (see, e.g.,~\cite{MR1271210}).
Indeed, it can be seen that {\sc FreeVecCon} is polynomially solvable on cycles. (We omit the easy proof.)
Block graphs and block-cactus graphs give further examples of graph classes with arbitrarily long induced paths
for which {\sc VecCon} is polynomially solvable (cf.~the discussion in~\cite{BHHM}).


For instances on arbitrary graphs with bounded requirements we can show the following result.

\begin{theorem}
The {\sc VecCon} problem can be solved in polynomial time if all requirements are at most $2.$
\end{theorem}

\begin{proof}
We work on the block tree $\cal T$ of $G$. For each leaf block $L$ with cut vertex $v$ such that for every vertex $u\in V(L)\setminus \{v\}$, we have
$r(u) \leq 1$, we delete from the graph the vertex set $V(L)\setminus \{v\}.$ We repeat this operation as long as possible.
Let ${\cal L}$ be the set of leaves of such a pruned block tree, which we denote~${\cal T}'$.
Our solution $S$ contains exactly one vertex $v_L$ from each block $L \in {\cal L}$ such that $v_L$ is not a cut vertex of the pruned graph,
unless $|V({\cal T}')| = 1$, in which case the pruned graph is biconnected, and the optimal solution is of size at most $2$.

Assuming
$|V({\cal T}')| > 1$, for each $L \in {\cal L}$ let ${\cal S}(L)$ be the set of blocks $B$ removed from ${\cal T}$ such that every block on the path from $L$ to $B$ in ${\cal T}$
has also been removed.
Clearly every deleted vertex $v$ of $G$ in a block of ${\cal S}(L)$ has requirement $1$ or $0$ and can be reached by a path
from the vertex selected in $S \cap L$.
For each remaining block $B$ we have that $B$ lies in ${\cal T}$ on a path between two leaf blocks of ${\cal T}'$ say $L_1, L_2 \in {\cal L}.$
Hence, each vertex in $B$ is reached by two disjoint paths, one starting in $S \cap L_1$ and one in $S \cap L_2.$

To argue optimality,
consider for each $L \in {\cal L}$ the set
$V_L$ of all vertices in $V(G)\setminus V(L)$ from a block in ${\cal S}(L).$ By Proposition 1, we have that
$V_L\cap S' \neq \emptyset$ for each feasible solution $S'$. Since for every two distinct $L_1, L_2 \in {\cal L}$, we have
$V_{L_1}\cap V_{L_2} = \emptyset,$ we conclude that any solution $S'$ must satisfy $|S'| \geq |{\cal L}|.$
\end{proof}

\section{Conclusion}\label{sec:concl}

We conclude with some questions for future research related to the \hbox{\sc VecCon} problem.
\begin{enumerate}
\item Recall that if in \hbox{\sc VecCon} each path is restricted to be of length exactly $1$, we get the {\sc Vector Domination} problem. If we set $r(v) = d(v)$ in the {\sc Vector Domination} problem, we get the {\sc Vertex Cover} problem.  This motivates the following question: What is the computational complexity of \hbox{\sc VecCon} on input instances such that $r(v) = d(v)$ for every vertex $v\in V$?
\item \begin{sloppypar}The \hbox{\sc VecCon} problem is polynomial for trees, block graphs, and split graphs, all subclasses of chordal graphs. What is the complexity of \hbox{\sc VecCon} in chordal graphs?\end{sloppypar}
\item Cographs and split graphs do not have induced paths on $4$ and $5$ vertices, respectively.
What is the complexity of \hbox{\sc VecCon} in the class of $P_k$-free graphs for $k\ge 5$?
In particular, is there a $k\ge 5$ such that \hbox{\sc VecCon} is {\sf NP}-complete for $P_k$-free graphs?

\item The \hbox{\sc VecCon} problem is {\sf APX}-hard if $R(V(G))= 4$ and polynomial if
$R(V(G))\le 2$. What is the complexity of \hbox{\sc VecCon} if $R(V(G))= 3$?
\item While
the \hbox{\sc VecCon} problem can be approximated in polynomial time within a factor of $\ln n +2$ on all $n$-vertex graphs~\cite{BHHM},
 Theorem~\ref{thm:APX} shows that there is no PTAS for \hbox{\sc VecCon} unless {\sf P = NP}.
Nevertheless, the exact (in-)approximability of \hbox{\sc VecCon} remains open, including the question from~\cite{BHHM}
whether \hbox{\sc VecCon} admits a polynomial-time constant-factor approximation algorithm on general graphs.
\end{enumerate}

Finally, we remark that {\sc VecCon} is fixed-parameter tractable with respect to solution size~\cite{KM,DL}.

\subsection*{Acknowledgements}

The second author is supported in part by the
Slovenian Research Agency, research program P$1$--$0285$ and research projects
J$1$--$4010$, J$1$--$4021$, J$1$--$5433$, and N$1$--$0011$:
GReGAS, supported in part by the European Science Foundation.

\bibliographystyle{abbrv}
\bibliography{VC-biblio}{}

\begin{thebibliography}{10}

\bibitem{Ackerman:2010}
E.~Ackerman, O.~Ben-Zwi, and G.~Wolfovitz.
\newblock Combinatorial model and bounds for target set selection.
\newblock {\em Theor. Comput. Sci.}, 411(44-46):4017--4022, 2010.

\bibitem{MR1756204}
P.~Alimonti and V.~Kann.
\newblock Some {APX}-completeness results for cubic graphs.
\newblock {\em Theoret. Comput. Sci.}, 237(1-2):123--134, 2000.

\bibitem{MR1734026}
G.~Ausiello, P.~Crescenzi, G.~Gambosi, V.~Kann, A.~Marchetti-Spaccamela, and
  M.~Protasi.
\newblock {\em Complexity and approximation}.
\newblock Springer-Verlag, Berlin, 1999.

\bibitem{BHHM}
E.~Boros, P.~Heggernes, P.~van~'t Hof, and M.~Milani{\v{c}}.
\newblock Vector connectivity in graphs.
\newblock {\em Networks}, 63(4):277--285, 2014.

\bibitem{Chopin:2014}
M.~Chopin, A.~Nichterlein, R.~Niedermeier, and M.~Weller.
\newblock Constant thresholds can make target set selection tractable.
\newblock {\em Theor. Comp. Sys.}, 55(1):61--83, 2014.

\bibitem{Cicalese201365}
F.~Cicalese, G.~Cordasco, L.~Gargano, M.~Milani\v{c}, and U.~Vaccaro.
\newblock Latency-bounded target set selection in social networks.
\newblock volume 7921 of {\em Lecture Notes in Computer Science}, pages 65--77,
  2013.

\bibitem{Cic-Mil-Vac:2013}
F.~Cicalese, M.~Milani\v{c}, and U.~Vaccaro.
\newblock On the approximability and exact algorithms for vector domination and
  related problems in graphs.
\newblock {\em Discrete Appl. Math.}, 161(6):750--767, 2013.

\bibitem{Dinh:2014}
T.~N. Dinh, Y.~Shen, D.~T. Nguyen, and M.~T. Thai.
\newblock On the approximability of positive influence dominating set in social
  networks.
\newblock {\em J. Comb. Optim.}, 27(3):487--503, 2014.

\bibitem{EvenTarjan}
S.~Even and R.~E. Tarjan.
\newblock Network flow and testing graph connectivity.
\newblock {\em SIAM J. Comput.}, 4:507--518, 1975.

\bibitem{MR1741406}
J.~Harant, A.~Pruchnewski, and M.~Voigt.
\newblock On dominating sets and independent sets of graphs.
\newblock {\em Combin. Probab. Comput.}, 8(6):547--553, 1999.

\bibitem{Kempe:2005}
D.~Kempe, J.~Kleinberg, and E.~Tardos.
\newblock Influential nodes in a diffusion model for social networks.
\newblock In {\em Proceedings of the 32Nd International Conference on Automata,
  Languages and Programming}, ICALP'05, pages 1127--1138, Berlin, Heidelberg,
  2005. Springer-Verlag.

\bibitem{KM}
S.~Kratsch and M.~Sorge.
\newblock On kernelization and approximation for the vector connectivity
  problem.
\newblock Manuscript, 2014. arXiv:1410.8819.

\bibitem{DL}
D.~Lokshtanov.
\newblock Personal communication, 2014.

\bibitem{zbMATH02582908}
K.~{Menger}.
\newblock {Zur allgemeinen Kurventheorie.}
\newblock {\em {Fundam. Math.}}, 10:96--115, 1927.

\bibitem{MR1828438}
B.~Mohar.
\newblock Face covers and the genus problem for apex graphs.
\newblock {\em J. Combin. Theory Ser. B}, 82(1):102--117, 2001.

\bibitem{MR1154589}
H.~Nagamochi and T.~Ibaraki.
\newblock A linear-time algorithm for finding a sparse {$k$}-connected spanning
  subgraph of a {$k$}-connected graph.
\newblock {\em Algorithmica}, 7(5-6):583--596, 1992.

\bibitem{MR0351881}
S.~Poljak.
\newblock A note on stable sets and colorings of graphs.
\newblock {\em Comment. Math. Univ. Carolinae}, 15:307--309, 1974.

\bibitem{MR1271210}
B.~Randerath and L.~Volkmann.
\newblock A characterization of well-covered block-cactus graphs.
\newblock {\em Australas. J. Combin.}, 9:307--314, 1994.

\end{thebibliography}

\end{document}